\documentclass[letterpaper, 10 pt, conference]{ieeeconf}
\IEEEoverridecommandlockouts
\usepackage{fix-cm}
\usepackage[utf8]{inputenc}
\usepackage{cite}
\usepackage{caption}
\usepackage{subcaption}
\usepackage{amsmath,amssymb,amsfonts}
\usepackage{algorithmic}
\usepackage{algorithm} 
\usepackage{algorithmic}
\usepackage{graphicx}
\usepackage{subcaption}
\usepackage[textsize=footnotesize]{todonotes}
\newtheorem{theorem}{Theorem}
\newtheorem{remark}{Remark}
\newtheorem{corollary}{Corollary}

\newtheorem{definition}{Definition}
\newtheorem{lemma}{Lemma}
\newtheorem{assumption}{\textbf{Assumption}}

\newtheorem{example}{\indent Example}

\def\and{\textrm{ and}}

\newcommand{\Sj}[1]{S^{(#1)}(\prtition_#1,\alpha_{#1})}
\newcommand{\hj}[1]{h^{(#1)}(x,\prtition_#1,\alpha_{#1})}

\newcommand{\Domain}{\mathcal{D}}

\newcommand{\PWA}{\mathrm{PWA}}

\newcommand{\genpwaparamM}{A}
\newcommand{\genpwaparamV}{a}

\newcommand{\Mode}{I}
\newcommand{\state}{x}

\newcommand{\pState}{X}

\newcommand{\Vertex}{\mathcal{F}_0}

\newcommand{\prtition}{\mathcal{P}}

\newcommand{\qState}{Z}

\newcommand{\vMatrix}{E}
\newcommand{\vVec}{e}

\newcommand{\Rprtition}{\mathcal R}

\newcommand{\Pindex}{\Mode(\prtition)}

\newcommand{\Int}[1]{\mathrm{Int}\left(#1 \right)}

\usepackage{hyperref}
\hypersetup{
colorlinks=true,
linkbordercolor=orange,
citebordercolor=blue,
citecolor=blue,
filecolor=magenta,
urlcolor=magenta,
urlbordercolor=blue
}

\author{Pouya Samanipour and Hasan A. Poonawala
\thanks{Pouya Samanipour and Hasan A. Poonawala are with the Department of Mechanical and Aerospace Engineering, University of Kentucky, Lexington, USA
{\tt\small\{samanipour.pouya,hasan.poonawala\}@uky.edu}. The corresponding author is Hasan A.Poonawala. 
This work is supported by the Department of Mechanical Engineering at the University of Kentucky.}}
\begin{document}

\title{Replacing $\mathcal{K}_\infty$ Function with Leaky ReLU in Barrier Function Design: A Union of Invariant Sets Approach for ReLU-Based Dynamical Systems}
\maketitle
\begin{abstract}
In this paper, a systematic framework is presented for determining piecewise affine ($\PWA$) barrier functions and their corresponding invariant sets for dynamical systems identified via Rectified Linear Unit (ReLU) neural networks or their equivalent $\PWA$ representations. A common approach to determining the invariant set is to use Nagumo's condition, or to utilize the barrier function with a class-$\mathcal{K}_\infty$ function. It may be challenging to find a suitable class-$\mathcal{K}_\infty$ function in some cases.
We propose leaky ReLU as an efficient substitute for the complex nonlinear $\mathcal{K}_\infty$ function in our formulation. Moreover,
we propose the Union of Invariant Sets (UIS) method, which combines information from multiple invariant sets in order to compute the largest possible $\PWA$ invariant set. 
The proposed framework is validated through multiple examples, showcasing its potential to enhance the analysis of invariant sets in ReLU-based dynamical systems. Our code is available at: \url{https://github.com/PouyaSamanipour/UIS.git}.
\end{abstract}

\section{Introduction}
The Rectified Linear Unit (ReLU) neural network (NN) is capable of identifying complex system dynamics, and ReLU-based neural networks have been successfully applied as NN-controllers for a variety of challenging control applications such as robot manipulation, autonomous driving~\cite{tan2018sim,rober2023backward}. In spite of their effectiveness, standard NN training lacks the safety and convergence guarantees provided by classical control methods. Moreover, ReLU-based controllers are inherently sensitive to small input perturbations, which can lead to unexpected and potentially unsafe behavior~\cite{yuan2019adversarial}. These limitations highlight the need for computational tools to ensure safety and reliability for ReLU NNs, particularly for applications in safety-critical domains~\cite{chang2019neural,dai2021lyapunov,dawson2022safe}.

Barrier functions (BFs) are a widely used approach for automatically guaranteeing safety in dynamical systems~\cite{ames2014control}. As safety filters in control systems, control barrier functions(CBF) are commonly used, but their effectiveness depends heavily on their robustness to model uncertainties~\cite{HAMDIPOOR2023100840}. The challenge of dynamical model uncertainty has been addressed by various sampling-based methods such as supervised learning, reinforcement learning (RL), and integration of CBFs in order to establish forward invariance sets and ensure the safety of uncertain systems~\cite{marvi2021safe,taylor2020learning}. Despite their success, in most applications it may be difficult to verify the learned safe set~\cite{dai2021lyapunov}. 

The piecewise affine ($\PWA$) nature of ReLU neural networks enables their representation as $\PWA$ functions~\cite{10108069}. The computation of invariant sets for both continuous and discrete-time $\PWA$ dynamical systems has been extensively investigated in the literature, with significant contributions addressing constrained systems and their invariant properties~\cite{blanchini1996constrained, rakovic2004invariant}. In~\cite{samanipour2024invariant}, we introduced a method to estimate invariant sets and construct $\PWA$ BFs for $\PWA$ dynamical systems. The BF found depends on the choice of the $\mathcal{K}_{\infty}$ function, $\alpha(x)$, which is restricted to be a linear function $\alpha(x) = \alpha x$ in~\cite{samanipour2024invariant}. Furthermore, their approach relies on trying different linear coefficients $\alpha$ and selecting the one that provides the largest possible invariant set. 
While trying various linear $\alpha(x)$ to obtain the largest possible set of invariants might be computationally inexpensive in lower dimensions, it might prove to be computationally demanding in higher-dimensional systems.

In the present work, we address the challenge of selecting an appropriate $\mathcal{K}_{\infty}$ function for barrier-based safety analysis. We propose using a Leaky ReLU function instead of a complex or highly nonlinear function for $\alpha(\cdot)$. This choice offers two key benefits: simplicity in form and a seamless path towards incorporating the non-smooth BF framework described in~\cite{glotfelter2017nonsmooth}. The non-smooth BF~\cite{glotfelter2017nonsmooth} utilizes multiple valid BFs to construct a new \emph{candidate} BF; the validity of this combined BF depends on the existence of an appropriate $\alpha(x)$ that satisfies essential safety conditions for these multiple BFs simultaneously. No method to find such $\alpha(x)$ is provided. By using the Leaky ReLU function as $\alpha(x)$, we are able to find such an $\alpha(x)$, so that our combined barrier function is valid. 
Building on this insight, we introduce a new method, the Union of Invariant Sets (UIS). The UIS method combines several $\PWA$ BFs derived from various linear functions $\alpha(x)$ for ReLU-based dynamical systems into a single  $\PWA$ BF. This integration may expand the overall invariant set while preserving safety guarantees and incurring no additional computational cost compared to~\cite{samanipour2024invariant}. 
\section{Preliminaries}
To begin, we briefly introduce the necessary notations and provide an overview of piecewise affine $\PWA$ functions.

\paragraph*{Notation}
Let $S$ be a set. The set of indices corresponding to the elements of $S$ is denoted by $\Mode(S)$. The convex hull, interior, and boundary $S$ are denoted by $\mathrm{conv}(S)$, $\mathrm{Int}(S)$, $\partial S$ respectively.
For a matrix $A$, $A^T$ denotes its transpose. 
As part of this paper, we present a definition of $\PWA$ dynamical systems on a partition. Consequently, the partition is defined as follows:
\begin{definition}
Throughout this paper, the partition $\mathcal{P}$ is a collection of subsets $\{\pState_i \}_{i \in \Pindex}$, where each $\pState_i$ represents a closed subset of $\mathbb{R}^n$ for all $i\in \Pindex$. In the partition $\mathcal{P}$, $\text{Dom}(\mathcal{P})=\cup_{i\in \Pindex} \pState_i$ and $\mathrm{Int}(\pState_i) \cap  \mathrm{Int}(\pState_j) = \emptyset$ for $i\neq j$.
\end{definition}

The other concept we need is the refinement of a partition. In mathematical terms, given two partitions $\prtition = \{Y_i\}_{i \in I}$ and $\Rprtition = \{ \qState_j\}_{j \in J}$ of a set $S = \mathrm{Dom}(\prtition) = \mathrm{Dom}(\Rprtition)$, we say that $\Rprtition$ is a refinement of $\prtition$ if $\qState_j \cap Y_i \neq \emptyset$ implies $\qState_j \subseteq Y_i$.

Furthermore, we use $ \dim(X_i)$ to denote the dimension of a cell $X_i$, where $i \in I(\prtition)$.
\subsection{Piecewise Affine Functions}
\label{sec:pwafun}
A piecewise affine function, denoted by $\PWA(x)$, is represented via an explicit parameterization based on a partition $\mathcal{P} = \{\pState_i\}_{i \in \Pindex}$. Each region in the partition corresponds to a set of affine dynamics, described by a collection of matrices $\mathbf{A}_\prtition = \{\genpwaparamM_i \}_{i \in \Pindex}$ and vectors $\mathbf{a}_\prtition = \{\genpwaparamV_i \}_{i \in \Pindex}$. The $\PWA$ function is defined as follows:
\begin{align}\label{eq:definePWAfun}
\PWA(x) &= \genpwaparamM_i \state + \genpwaparamV_i, \quad \text{if } \state \in \pState_i, 
\end{align}
where the region (cell) $\pState_i$ is described by
\begin{align}\label{eq:H rep}
  \pState_i &= \{x \in \mathcal{R}^n \colon \vMatrix_i \state + \vVec_i  \succeq 0\}, 
\end{align}
with matrices $E_i \in \mathbb{R}^{n_{hi} \times n}$ and vectors $e_i \in \mathbb{R}^{n_{hi}}$, defining the boundary hyperplanes of $\pState_i$. Here, $n_{h_i}$ denotes the number of hyperplanes in cell $X_i$. 
\begin{assumption}\label{ass:assumption bounded polytopes}
In this work, we assume that all partition cells are bounded polytopes. 
Consequently, the vertex representation of the $\PWA$ dynamics is valid. 
Specifically, each cell $X_i$ in the partition can be represented as the convex hull of its set of vertices, $\Vertex(X_i)$, as follows:
\begin{equation}
X_i = \mathrm{conv}\{\Vertex(X_i)\}.
\end{equation}
\end{assumption}
\section{Invariant Set Estimation}
To describe the invariant set estimation algorithm, we must first define the concept of the set of invariance.
\subsection{Forward Invariance}
To explain the concepts of forward invariant, consider the following nonlinear dynamics:
\begin{equation}\label{eq:general nl}
\dot{x} = f_{cl}(x).
\end{equation}
Consider $f_{cl}$ as locally Lipschitz continuous within the domain $\Domain \subseteq \mathbb{R}^n$. Based on this assumption, for any initial condition $x_0 \in \Domain$, there is a time interval $I(x_0) = [0, \tau_{\text{max}})$ within which a unique solution $x(t)$ to~\eqref{eq:general nl} exists, fulfilling the differential equation~\eqref{eq:general nl} and the initial condition $x_0$~\cite{ames2019control}. 
\begin{definition}[Forward invariant set~\cite{ames2019control}]
Let us define a super-level set $S$ corresponding to a continuously differentiable function $h:\Domain \subset \mathbb{R}^n \rightarrow \mathbb{R}$ for the closed-loop system $f_{cl}$ given by equation~\eqref{eq:general nl} as follows:
\begin{align}\label{eq:BF Domain}
S =& \{x \in \Domain : h(x) \geq 0\},
\end{align}
where set $S$ is considered forward-invariant if the solution $x(t)$ remains in $S$ for all $t \in I(x_0)$ for every $x_0$. 
\end{definition}
\begin{definition}[Barrier function~\cite{ames2019control}]\label{Def: Barrier function}
Let $S \subset \Domain \subseteq \mathbb{R}^n$ represent the superlevel set of a continuously differentiable function $h(x)$. It is said that $h(x)$ is a BF if there exists an extended class $\mathcal{K}_\infty$ function $\alpha(x)$ in which:
\begin{equation}\label{eq:barrier general}
L_{f_{cl}} h(x) \geq -\alpha(h(x)), \quad \text{for all } x \in D,
\end{equation}
where lie derivative of $h(x)$ along the closed loop dynamics $f_{cl}$ is denoted by $L_{f_{cl}} h(x)$. 
Equation~\eqref {eq:barrier general} makes set $S$ an asymptotically set in $\Domain$.
\end{definition}
A key challenge in finding the BF using~\eqref{eq:barrier general} is choosing the $\mathcal{K}_\infty$ function $\alpha(x)$. A method for selecting $\alpha(x)$ will be discussed in the following section.
\subsection{Leaky ReLU: A practical choice for \texorpdfstring{$\alpha(x)$}{alpha(x)} in BF design}
Definition~\ref{Def: Barrier function} shows that if there exists an $\alpha(x)$ satisfies constraint~\eqref{eq:barrier general}, then $h(x)$ is a BF, and $S$ remains forward invariant and asymptotically stable in $\Domain$. Finding a suitable nonlinear $\alpha(x)$ is often challenging and computationally expensive because it often involves searching through a large function space. It is common practice to simplify the analysis using a linear $\alpha(x)$,~\cite{samanipour2024invariant}. However, this methodology may impose limitations on the search for a larger invariant set. 

To address these limitations, we propose a Leaky ReLU function as a practical and flexible substitute. 
Unlike approaches that rely on complex $\mathcal{K}_\infty$ functions or trial-and-error selection of linear $\alpha(x)$, the Leaky ReLU function provides an efficient alternative with only a few tunable parameters, enabling more effective barrier function construction for dynamics~\eqref{eq:general nl}.

\begin{theorem}\label{th:leaky}
Let $S \subset \Domain \subseteq \mathbb{R}^n$ represent the super-level set, as defined in Equation~\eqref{eq:BF Domain}, of a continuously differentiable function $h(x)$ with respect to the closed-loop dynamic~\eqref{eq:general nl}.
If $h(x)$ is a bounded function such that $h_{\min} \leq h(x) \leq h_{\max}$ for $x \in \Domain$, then the followings are equivalent:
\begin{enumerate}
    \item[(a)] There exists an extended class $\mathcal{K}_\infty$ function $\alpha(x)$ such that the barrier constraint~\eqref{eq:barrier general} is satisfied, rendering $S$ invariant (i.e., $h(x)$ is a valid BF).
   \item[(b)] There exists an extended class $\mathcal{K}_\infty$ function  $\overline{\alpha}(x) = \alpha_m \sigma_{\left(\frac{\alpha_1}{\alpha_m}\right)}(x)$ with constants $0 < \alpha_1 \leq \alpha_m$, where $\sigma_{\left(\frac{\alpha_1}{\alpha_m}\right)}$ is a Leaky ReLU function defined as 
    \begin{equation}\label{eq:Leaky relu}
        \sigma_{\left(\frac{\alpha_1}{\alpha_m}\right)}(x) =
        \begin{cases} 
            x, & \text{if } x \geq 0, \\
            \left(\frac{\alpha_1}{\alpha_m}\right)x, & \text{if } x < 0,
        \end{cases}
    \end{equation}
    such that condition~\eqref{eq:barrier general} is satisfied,rendering $S$ invariant (i.e., $h(x)$ is a valid BF). 
\end{enumerate}
\end{theorem}
\begin{proof}
The Leaky ReLU function $\sigma_{\left(\frac{\alpha_1}{\alpha_m}\right)}$ in~\eqref{eq:Leaky relu} is continuous, strictly increasing, and satisfies $\sigma_{\left(\frac{\alpha_1}{\alpha_m}\right)}(0)=0$. Thus, if \textit{(b)} holds, it immediately implies \textit{(a)}.

Now, suppose (a) is true. Since $h(x)$ is bounded on $\Domain$ with $h_{\min} \le h(x) \le h_{\max}$, and $\alpha(x)$ is continuous, the composition $\alpha\bigl(h(x)\bigr)$ is also bounded.

\textbf{Case 1:} $0 < h(x) \le h_{\max}$.  
Because $h(x)$ and $\alpha\bigl(h(x)\bigr)$ remain finite, the ratio $\frac{\alpha\bigl(h(x)\bigr)}{h(x)}$
is bounded above. Hence, there exists $\alpha_m > 0$ such that
\begin{equation}\label{eq:proof_upper}
  \alpha_m \;\ge\; \frac{\alpha\bigl(h(x)\bigr)}{h(x)},
  \quad
  \forall\, x \in \Domain \text{ with } 0 < h(x)\le h_{\max}.
\end{equation}
Adding $\dot{h}(x)$ to both sides of \eqref{eq:proof_upper} and using $\dot{h}(x) + \alpha\bigl(h(x)\bigr)\ge 0$, we obtain
\begin{align}
  \dot{h}(x) + \alpha_m\,h(x) 
  &\;\ge\; \dot{h}(x) + \alpha\bigl(h(x)\bigr)\geq 0 
\end{align}

\textbf{Case 2:} $h_{\min} \le h(x) < 0$.  
Similarly, $\frac{\alpha\bigl(h(x)\bigr)}{h(x)}$ is bounded below. Thus there exists $\alpha_1 > 0$ such that
\begin{equation}\label{eq:proof_lower}
  \alpha_1 \;\le\; \frac{\alpha\bigl(h(x)\bigr)}{h(x)},
  \quad
  \forall\, x \in \Domain \text{ with } h_{\min}\le h(x)<0.
\end{equation}
Again, adding $\dot{h}(x)$ to both sides of \eqref{eq:proof_lower} yields
\begin{align}
  \dot{h}(x) + \alpha_1\,h(x)
  &\;\ge\; \dot{h}(x) + \alpha\bigl(h(x)\bigr)\geq 0.
\end{align}

From these two cases, we see that we can define $\alpha(x)$ as described in~\eqref{eq:Leaky relu}. Therefore, if (a) holds, (b) must also hold. This completes the proof.
\end{proof}
Theorem~\ref{th:leaky} provides insight into how BFs can be constructed by tuning a Leaky ReLU with two parameters rather than looking for a complex $\mathcal{K}_\infty$ function. For less conservative behavior, $\alpha_m$ could be set to a large value and $\alpha_1$ to a small value. The boundedness of $h(x)$ imposes no restrictions if $\Domain$ is a compact set.
\subsection{\texorpdfstring{$\PWA$}{PWA} barrier function}
The main goal of this paper is to estimate the invariant set for the dynamical systems identified using ReLU NN or its equivalent $\PWA$ dynamical systems as follows:
\begin{equation}\label{eq:pwa dynamic}
    \dot x=\PWA(x),
\end{equation}
where the $\PWA$ function is defined on a pre-selected domain $\Domain$. In~\cite{samanipour2024invariant}, we proposed an approach to estimate the invariant set for dynamical systems~\eqref{eq:pwa dynamic} with linear $\alpha(x)$. However, using linear $\alpha(x)$ might be challenging. Before discussing how leaky ReLU $\alpha(x)$ can address this challenge, it is necessary to take a look at our work~\cite{samanipour2024invariant}.

Our work in~\cite{samanipour2024invariant} consists of parameterizing the BF as a $\PWA$ function on the same partition as the dynamical systems~\eqref{eq:pwa dynamic}. The BF for the cell $X_i$ can be parameterized as follows:
\begin{equation}\label{eq: BF PWA}
h_i(x)=s_i^Tx+t_i \quad \text{for} \quad i \in I(\prtition) 
\end{equation}
where $s_i\in \mathcal{R}^n$ and $t_i\in \mathcal{R}$.
$h_i(x)$ is a continuous and differentiable function within the interior of the cell. 
\begin{assumption}\label{assumption:derivate}
Let's consider a cell $X_j$ with local dynamic $\dot{x} = A_jx + a_j$ and a candidate BF $h_j(x) = s_j^T x + t_j$. Then the derivative of the BF along the dynamic solution at $\mathbf{x}'\in \Int{X_j}$ can be obtained as follows:
\begin{equation}\label{eq:BF derivative definition}
\dot{h_j}(\mathbf{x}') = s_j^T(A_jx'+a_j).
\end{equation}
For more details please see~\cite{10108069,samanipour2024invariant}.
\end{assumption}
 
In~\cite{samanipour2024invariant}, we constructed an optimization problem to search for a $\PWA$ invariant set over the domain $\Domain$. To describe the optimization problem, we must first establish useful index sets. Since our algorithm depends on the partition, we denote all these sets as functions of the partition.
We define the index set $I_{\partial \Domain}$ as follows:  
\begin{equation}
I_{\partial \Domain}(\prtition) = \{ i \in I \colon X_i \cap \partial \Domain \neq \emptyset \},
\end{equation}
where $I_{\partial \Domain}(\prtition)$ represents the indices of cells that have a non-empty intersection with the boundary of the $\Domain$.
To determine whether a vertex belongs to the $\partial \Domain$ or the $\Int{\Domain}$, we introduce the following sets:
\begin{align}
I_b(\prtition) &= \{ (i, k) \colon v_k \in \partial \Domain, \, i \in I_{\partial \Domain}, \, v_k \in \Vertex(X_i) \}, \\
I_{int}(\prtition) &= \{ (i, k) \colon v_k \notin \partial \Domain, \, i \in I(\prtition), \, v_k \in \Vertex(X_i) \}.\nonumber
\end{align}
Here, $I_b(\prtition)$ and $I_{\mathrm{Int}}(\prtition)$ are sets of ordered pairs where the first element corresponds to the cell index $i$, and the second element corresponds to the vertex index $k$. Specifically, $I_b(\prtition)$ identifies the vertices on the boundary of $\Domain$ and their associated cells, while $I_{\mathrm{Int}}(\prtition)$ identifies the vertices in the interior of the domain partition $\prtition$ and their associated cells.

We constructed a linear optimization in~\cite{samanipour2024invariant} as follows to find a certified invariant set with $\alpha(x)=\alpha x$ where $\alpha>0$. 
\begin{subequations}\label{eq:opt UIS}
\begin{align}
&\min_{ s_i, t_i,\tau_{\mathrm{Int}_{i}},\tau_{b_{i}}}  \quad  \sum_{i=1}^{M}\tau_{b_{i}}+\sum_{i=1}^{N}\tau_{\mathrm{Int}_{i}},\label{eq:cost_function}\\
&\text{Subject to:} \nonumber\\
&h_i(v_k)-\tau_{b_i}\leq -\epsilon_1, \quad \forall (i,k) \in I_{b}(\mathcal{P}),\label{eq:NB}\\
&h_i(v_k)+\tau_{\mathrm{Int}_i}\geq\epsilon_2, \quad \forall (i,k) \in I_{\mathrm{Int}}(\mathcal{P}),\label{eq:PI} \\
&\dot{h_i}(v_k)+\alpha h_i(v_k)\geq \epsilon_3, \quad \forall (i,k) \in I_{\Domain}(\mathcal{P}),\label{eq:Nagumo}\\
&h_i(v_k)=h_j(v_k), \quad \forall v_k \in \Vertex(X_i) \cap \Vertex(X_j),\label{eq:continuity}\\
&\tau_{b_i},\tau_{\mathrm{Int}_{i}}\geq 0,\label{eq:PS}
\end{align}    
\end{subequations}

where $\epsilon_1,\epsilon_2,\epsilon_3>0$. In~\eqref{eq:cost_function}, $M$ and $N$ denote the number of elements in $I_{b}(\mathcal{P})$ and $I_{\mathrm{Int}}(\mathcal{P})$, respectively. Constraint~\eqref{eq:NB} ensures that vertices in $I_{b}(\mathcal{P})$ are excluded from the invariant set, while constraint~\eqref{eq:PI} guides the search algorithm to include all points in $I_{\mathrm{Int}}(\mathcal{P})$. 

These constraints,~\eqref{eq:NB} and~\eqref{eq:PI}, are relaxed, as some vertices in $I_{\mathrm{Int}}(\mathcal{P})$ may not be contained within the invariant set, or the search algorithm may encounter feasibility issues for a given partition. Following the parametrization, constraint~\eqref{eq:Nagumo} is derived from the general constraint~\eqref{eq:barrier general}. The feasibility of this optimization problem is analyzed in~\cite{samanipour2024invariant}. For further details, please refer to~\cite{samanipour2024invariant}.
\begin{remark}
We must note that the optimization problem~\eqref{eq:opt UIS} is always feasible, however, its solution might not be a certified invariant set. In case the $\sum_{i=1}^{M}\tau_{b_i}\neq 0$ in the optimization problem~\eqref{eq:opt UIS}, the current partition does not justify a certified invariant set since certain constraints on the boundary cells,~\eqref{eq:NB}, have not been met. In order to resolve this issue, all boundary cells where the slack variable $ \tau_{b_i}$ is non-zero must be refined. Refinement should also be considered for interior cells with non-zero slack variables, $\tau_{\mathrm{Int}_i}$. 
In this work, we utilized the vector field refinement approach presented in~\cite{10313502}.    
\end{remark}
In~\cite{samanipour2024invariant}, we employed a bisection method to obtain a set of parameters $\underline{\alpha} \;=\; \{\alpha_1, \alpha_2, \dots, \alpha_m\},$
each yielding an associated invariant set. We then retained only the invariant set corresponding to the parameter $\alpha_k \in \underline{\alpha}$ that resulted in the largest invariant set, discarding all other sets obtained from the remaining parameters. Although this strategy was effective it is not computationally efficient.

As described earlier, Leaky ReLUs $\alpha(x)$ may address the challenge of using linear $\alpha(x)$ in~\cite{samanipour2024invariant}.
However, due to the constraint~\eqref{eq:Nagumo}, the resulting formulation requires mixed-integer programming, which results in high computational costs. Another potential solution is to take advantage of the concept of non-smooth BFs developed in~\cite{glotfelter2017nonsmooth}, where multiple invariant sets can be merged. The primary challenge in~\cite{glotfelter2017nonsmooth} is determining a suitable $\alpha(x)$ that simultaneously satisfies all the necessary barrier-function conditions, which can be difficult in practice. 
 In the next step, we will describe a new algorithm for estimating the invariant set for dynamical systems~\eqref{eq:pwa dynamic}.
\subsection{Union of Invariant Sets(UIS)}\label{sec:UIS}
This section addresses the challenge of finding suitable $\alpha(x)$ by integrating non-smooth BF with leaky ReLU $\alpha(x)$. This combination aims to produce a possibly larger invariant set compared to~\cite{samanipour2024invariant} without any additional computational cost. To demonstrate how non-smooth BFs can be used with the Leaky ReLU $\alpha(x)$, we must first modify our notation. 

As mentioned earlier, the solution to the optimization problem~\eqref{eq:opt UIS} depends on both $\mathcal{P}$ and $\alpha(x)$. Henceforth, we will use the notation $S^{(m)}(\mathcal{P}_m, \alpha_m)$ and $h^{(m)}(x, \mathcal{P}_m, \alpha_m)$ to reflect this dependency. The rest of the paper assumes $\alpha(x) = \alpha_m x$ in $h^{(m)}(x, \mathcal{P}_m, \alpha_m)$ unless otherwise stated. 
\begin{lemma}\label{lemma:UIS}
Consider the dynamical system given by~\eqref{eq:pwa dynamic} with an equilibrium at the origin, defined on a compact set $\Domain$.
Let $\underline{\alpha} = [\alpha_{min},\dots, \alpha_{max}]$ be a set of $m$ parameters ordered in ascending order. 
Suppose that the corresponding invariant sets, obtained through the optimization problem~\eqref{eq:opt UIS} are denoted by $\Sj{i}$ for $i=1,\dots,m$ with $\alpha(x)=\alpha_ix$.
Then, the following results hold:
\begin{enumerate}
    \item There exists an invariant set $\overline{S}$ with respect to the dynamical system~\eqref{eq:pwa dynamic}, where \( \overline{S} \) is given by:
    \begin{equation}\label{eq:union of of inv}
    \overline{S} = \bigcup_{i=1}^m \Sj{i}.    
    \end{equation}
    \item The set $\overline{S}$ is rendered asymptotically stable in $\Domain$ with the following BF and class-$\mathcal{K}_{\infty}$ function:
    \begin{equation}\label{eq:h UIS}
    \overline{h}(x, \overline{\prtition}, \overline{\alpha}) = \max_{i} \{ \hj{i} \},
    \end{equation}
    where $\overline{\alpha}(x)$ is defined as:
    \begin{align}
        \overline{\alpha}(x) = \alpha_{max} \sigma_{\left(\frac{\alpha_{min}}{\alpha_{max}}\right)}(x),
    \end{align}
    and $\sigma_{\left(\frac{\alpha_{min}}{\alpha_{max}}\right)}(x)$ denotes the Leaky ReLU function as defined in~\eqref{eq:Leaky relu}. The partition $\overline{\prtition}$ is the product partition defined as:
    \begin{equation}
        \overline{\mathcal{P}} = \mathcal{P}_1 \times \mathcal{P}_2 \times \dots \times \mathcal{P}_m, 
    \end{equation}
    where:
    \begin{equation}
     \prtition^*=\prtition_1\times\prtition_2=\{X_k\}_{k\in I(\prtition^*)},\nonumber   
    \end{equation}
    \begin{equation}
    X_k=\{X_i\cap X_j:i\in I(\prtition_1) ,j\in I(\prtition_2),dim(X_k)=n\}.\nonumber    
    \end{equation}
\end{enumerate}
\end{lemma}
\begin{proof}
The first part of the proof can be proved using the contradiction.
Assume, for contradiction, that there exists an initial condition $x_0 \in \overline{S}$ such that the trajectory of the system leaves $\overline{S}$ at some future time.

By definition of $\overline{S}$, for any $x_0 \in \overline{S}$, there exists $i \in \{1, 2, \dots, m\}$ such that $x_0 \in \Sj{i}$. Since $\Sj{i}$ is invariant by definition, any initial condition $x_0 \in \Sj{i}$ implies that the trajectory will remain within $\Sj{i}$ for all future time steps.

However, this contradicts our assumption that the trajectory leaves $\overline{S}$. Therefore, for any $x_0 \in \overline{S}$, the trajectory of the system will always remain within one of the sets $\Sj{i}$, which are invariant by construction.

To prove the second part, note that $\overline{h}(x, \overline{\mathcal{P}}, \overline{\alpha})$ is a continuous $\PWA$ BF, since all $\hj{i}$ for $i = 1, \dots, m$ are continuous and $\PWA$. To make the definition of this new $\PWA$ BF compatible with~\eqref{eq:definePWAfun}, we defined $\overline{\prtition}$ to ensure that in each cell, we have only one affine function $\overline{h}_i(x,\overline{\prtition},\overline{\alpha})=s_i^Tx+t_i$ for $x\in X_i(\overline{\prtition})$.
 
To prove that the candidate BF is valid, we first need to show that $\overline{h}_i(x,\overline{\prtition},\overline{\alpha})\geq 0$ for $x\in \overline{S}$ and $\overline{h}_i(x,\overline{\prtition},\overline{\alpha})< 0$ for $x \notin \overline{S}$. By definition of the $\overline{h}(x, \overline{\prtition} , \overline{\alpha}) $ in~\eqref{eq:h UIS}, we know
\begin{equation}\label{eq:proof pos}
    \overline{h}(x, \overline{\prtition}, \overline{\alpha})\geq \hj{i} \quad i=1,\dots,m,\hspace{0.2em} x\in \Domain. 
\end{equation} 
Moreover, by definition of $\overline{S}$, for any $x \in \overline{S}$, there exists $i \in \{1, 2, \dots, m\}$ such that $x \in \Sj{i}$ and as a result $\hj{i}\geq 0$. Therefore, with respect to~\eqref{eq:proof pos}, and due to the existence of $\hj{i} \geq 0$ for all $x \in \overline{S}$, we can conclude that $\overline{h}(x, \overline{\mathcal{P}}, \overline{\alpha}) \geq 0$ for $x \in \overline{S}$.
 Due to brevity, the proof of $\overline{h}(x,\overline{\prtition},\overline{\alpha})<0$ for $x\notin \overline{S}$ is eliminated.  

Now, we need to prove that condition~\eqref{eq:barrier general} holds for $\overline{h}(x,\overline{\prtition},\overline{\alpha})$ with respect to the dynamical system~\eqref{eq:pwa dynamic} for $x\in \overline{S}$. According to Theorem~\ref{th:leaky}, if the condition in~\eqref{eq:barrier general} holds for $\overline{h}(x,\overline{\prtition},\overline{\alpha})$ with an arbitrary $\overline{\alpha}(x)$, then $\overline{\alpha}(x)$ can be substituted with a Leaky ReLU function. Give that, let us assume $\overline{\alpha}(x)=\alpha_{max}\sigma_{(\frac{\alpha_{min}}{\alpha_{max}})}(x)$.
Let us consider $\mathbf{x} \in \Int{\overline{S}}$ such that $\overline{h}(\mathbf{x}, \overline{\mathcal{P}}, \overline{\alpha})$ is differentiable with respect to $\mathbf{x}$. To prove by contradiction, suppose equation~\eqref{eq:barrier general} does not hold for $\mathbf{x}$. By the definition of $\overline{h}(x,\overline{\prtition},\overline{\alpha})$ in~\eqref{eq:h UIS},   
$\overline{h}(\mathbf{x}, \overline{\mathcal{P}}, \overline{\alpha}) = h^{(k)}(\mathbf{x}, \mathcal{P}_k, \alpha_k) $ for $\mathbf{x}\in \Int{\overline{S}}$, where $k$  can be any value $1 \leq k \leq m$. Due to the fact that~\eqref{eq:barrier general} holds for $h^{(k)}(x, \mathcal{P}_k, \alpha_k)$ for $x \in \Domain$, we can conclude that:
\begin{equation}\label{eq:bf proof der}
\dot{h}^{(k)}(\mathbf{x}, \mathcal{P}_k, \alpha_k) + \alpha_k h^{(k)}(\mathbf{x}, \mathcal{P}_k, \alpha_k) \geq 0.    
\end{equation}
Also, the following holds because $\alpha_{max}\geq \alpha_k$ and $\overline{h}(\mathbf{x},\overline{\prtition},\overline{\alpha}))=h^{(k)}(\mathbf{x},\prtition_k,\alpha_k)> 0$ for $\mathbf{x}\in \Int{\overline{S}}$:
\begin{equation}\label{eq:am>ak}
    \alpha_{max} \overline{h}(\mathbf{x}, \overline{\prtition}, \overline{\alpha}) \geq \alpha_k h^{(k)}(\mathbf{x}, \mathcal{P}_k, \alpha_k).
\end{equation}
With respect to the assumption that $\mathbf{x}$ is a differentiable point we can conclude that $\dot{\overline{h}}(\mathbf{x}, \overline{\prtition}, \overline{\alpha})=\dot{h}^{(k)}(\mathbf{x}, \prtition_k, \alpha_k)$ and by adding it 
to both sides of inequality~\eqref{eq:am>ak}:
\begin{align}\label{eq:proof UIN interior}
    \dot{\overline{h}}(\mathbf{x}, \overline{\prtition}, \overline{\alpha}) +& \alpha_{max} \overline{h}(\mathbf{x}, \overline{\prtition}, \overline{\alpha})\\ &\geq \dot{h}^{(k)}(\mathbf{x}, \mathcal{P}_k, \alpha_k) + \alpha_k h^{(k)}(\mathbf{x}, \mathcal{P}_k, \alpha_k) \geq 0\nonumber.
\end{align}
which is a contradiction to the assumption that~\eqref{eq:barrier general} does not hold for $\overline{h}(\mathbf{x},\overline{\prtition},\overline{\alpha})$ for $\mathbf{x}\in \Int{\overline{S}}$. As a result, we can conclude for all the differentiable points in the $\Int{\overline{S}}$, $$\dot{\overline{h}}(x, \overline{\prtition}, \overline{\alpha}) + \overline{\alpha}\overline{h}(x, \overline{\prtition}, \overline{\alpha})\geq 0.$$

The same reasoning can be applied to points outside the invariant set $\overline{S}$. For differentiable points, such as $\mathbf{x}$, outside the invariant set, where $\overline{h}(\mathbf{x}, \overline{\mathcal{P}}, \overline{\alpha}) = h^{(k)}(\mathbf{x}, \mathcal{P}_k, \alpha_k) < 0$, let us assume, for contradiction, that~\eqref{eq:barrier general} does not hold. Since $\overline{h}(\mathbf{x}, \overline{\mathcal{P}}, \overline{\alpha}) < 0$, we have $\overline{\alpha}(\mathbf{x}) = \alpha_{min}$, where $\alpha_{min}\leq \alpha_k$ and $h(\mathbf{x}, \mathcal{P}_k, \alpha_k)< 0$. Therefore,
\begin{equation}
    \alpha_{min} \overline{h}(\mathbf{x}, \overline{\mathcal{P}}, \overline{\alpha}) \geq \alpha_k h^{(k)}(\mathbf{x}, \mathcal{P}_k, \alpha_k). \nonumber
\end{equation}
Adding \(\dot{\overline{h}}(\mathbf{x}, \overline{\mathcal{P}}, \overline{\alpha}) = \dot{h}^{(k)}(\mathbf{x}, \mathcal{P}_k, \alpha_k)\) to both sides of this inequality, we maintain the validity of~\eqref{eq:bf proof der}. Thus,
\begin{align}\label{eq:proof UIN exterior}
    \dot{\overline{h}}(\mathbf{x}, \overline{\mathcal{P}}, \overline{\alpha}) + \alpha_{min} \overline{h}(\mathbf{x}, \overline{\mathcal{P}}, \overline{\alpha}) \geq 0,
\end{align}
which contradicts our assumption. This allows us to conclude that equation~\eqref{eq:barrier general} holds for points outside the invariant set. 

The non-smooth BFs for differential inclusion are discussed in detail in ~\cite{glotfelter2017nonsmooth}. Proposition 2 in~\cite{glotfelter2017nonsmooth} demonstrates that~\eqref{eq:h UIS} is valid for non-differentiable points.  For the sake of brevity, we will skip this part. 

Consequently, $\overline{h}(x, \overline{\mathcal{P}}, \overline{\alpha})$ is a certified barrier function for the system described by~\eqref{eq:pwa dynamic}.
\end{proof}
\begin{remark}
\label{rem:UIS-extension}
Lemma~\ref{lemma:UIS} indicates that when multiple valid BFs exist, they can be combined to form a single BF using~\eqref{eq:h UIS}. It is shown here for $\PWA$, but this approach can also be applied to more general nonlinear dynamics on compact sets, as described in~\eqref{eq:general nl}, even though those broader applications are beyond the scope of this work.
\end{remark}
 
Building on Lemma~\ref{lemma:UIS}, we introduce the Union of Invariant Sets (UIS) method for constructing invariant sets. Instead of identifying a single Leaky ReLU function $\alpha(x)$ that maximizes the invariant set, UIS forms a unified invariant set by taking the union of multiple invariant sets, each obtained by solving~\eqref{eq:opt UIS} with a distinct linear $\alpha(x)$. This may result in a broader coverage of the state space while maintaining all interior points.  
The final invariant set is characterized by the barrier function~\eqref{eq:h UIS} and a Leaky ReLU $\alpha(x)$, integrating information from various valid BFs with linear $\mathcal{K}_\infty$ functions systematically. The complete procedure is detailed in Algorithm~\ref{alg:union of invariant sets}.
\begin{corollary}\label{cor:UIS}
Consider the set of parameters $\underline{\alpha} = \{\alpha_1, \alpha_2, \dots, \alpha_m\}$, and suppose $\alpha_k \in \underline{\alpha}$ yields the largest invariant set $S' = \Sj{k}$ through the optimization problem~\eqref{eq:opt UIS}. Then, the following holds for~$\overline{S}$,~\eqref{eq:union of of inv}: 
    \begin{equation}
        S' \subseteq \overline{S}.
    \end{equation}
\end{corollary}
\begin{proof}
    The result follows directly from the construction of the UIS algorithm.
\end{proof}
As a result of Corollary~\ref{cor:UIS}, it is important to note that UIS does not guarantee a larger forward invariant set in all cases; for instance, solutions such as $S(\mathcal{P}_4, \alpha_4)$ in Figure~\ref{fig:UIS process} do not contribute any new points to the UIS algorithm.
However, in comparison to the previous approach proposed in~\cite{samanipour2024invariant}, this method offers the advantage of incurring no additional computational expense, while simultaneously allowing for the determination of a potentially larger invariant set. 
\begin{algorithm}[tb]
    \begin{algorithmic}
    \REQUIRE  $\PWA(x)$ dynamic~\eqref{eq:pwa dynamic} and $\underline{\alpha}=\{\alpha_1,\dots,\alpha_m\}.$
    \FOR{$j=1,2,\dots,m$}
        \STATE{1-$\alpha(x)=\alpha_j\times x, \alpha_j\in \underline{\alpha}$}
        \STATE{2- Solve optimization problem~\eqref{eq:cost_function} for the $\PWA$ dynamic with $\alpha(x)$}
        \WHILE{$\sum_{i=1}^{N} \tau_{b_i} \neq 0$}  
            \FOR{each $i$ such that $\tau_{b_i} \neq 0$ and $\tau_{int_i} \neq 0$} 
                \STATE{Refine cells as described in~\cite{10313502}.}
            \ENDFOR 
            \STATE{Search for a certified invariant set using optimization problem~\eqref{eq:cost_function}
            for the $\PWA$ dynamics~\eqref{eq:pwa dynamic} with $\alpha(x)$ over the refined partition of the domain.}
        \ENDWHILE
    \ENDFOR
    \RETURN The Union of Invariant Sets~\eqref{eq:union of of inv} and its corresponding $\PWA$ BF~\eqref{eq:h UIS} with Leaky ReLU $\alpha(x)$. 
    \end{algorithmic}
    \caption{Union of Invariant sets(UIS)~\ref{sec:UIS} for $\PWA$ dynamics~\eqref{eq:pwa dynamic}}
    \label{alg:union of invariant sets}    
\end{algorithm}
\section{Results and Simulations}
All computations are implemented using Python 3.11 and on a computer with a 2.1 GHz processor and 8 GB RAM. A tolerance of $10^{-6}$ is used to determine if a number is nonzero. 
Moreover, the examples employ a tolerance of ${\epsilon_1=\epsilon_2=\epsilon_3=10^{-4}}$. 
\begin{example}[Inverted Pendulum~\cite{samanipour2024invariant}]\label{example:IP}
The inverted pendulum system can be modeled as follows:
\begin{align}
\dot{x}_1 &= x_2 \nonumber\\
\dot{x}_2 &= \sin(x_1) + u \nonumber
\end{align}
where $x_1$ represents the pendulum's angle, and $x_2$ is its angular velocity. The control input $u$, defined as $u = -3x_1 - 3x_2$, is constrained by a saturation limit between the lower bound of $-1.5$ and the upper bound of $1.5$.
The system's domain, $\Domain$, is given by $\Domain = \pi - ||x||_\infty \geq 0$, where $||x||_\infty$ represents the infinity norm. We utilize a ReLU neural network with a single hidden layer with eight neurons to approximate the right-hand side of the dynamics. The UIS estimates the invariant set with $\alpha=[0.025,0.05,0.06]$. The results are shown in Fig~\ref{fig:UIS}. The invariant set obtained with UIS is larger than the invariant set obtained with a linear $\alpha(x)$ as can be seen in Fig~\ref{fig:UIS}.
\end{example}
\begin{example}[Barrier Certificate Verification]
\label{comparison example}
Consider the following continuous system from~\cite{zhao2021synthesizing}:
\begin{align*}
    \dot{x}_1 &= x_1^2 + x_1 x_2 + x_1, \\
    \dot{x}_2 &= x_1 x_2 + x_2^2 + x_2,
\end{align*}
defined over the state space $\Domain = \{ x : \|x\|_\infty \le 2 \}$. The objective is to ensure that any trajectory that originates in the set $I_m = \{ x : 0.5 \leq x_1, x_2 \leq 1 \}$ never enters the unsafe region $U_m = \{ x : -1 \leq x_1, x_2 \leq 0 \}$.

The dynamic is identified using a ReLU NN with 20 neurons. UIS, as described in Algorithm~\ref{alg:union of invariant sets}, utilized $\alpha=[0.1,0.5]$ to obtain a new invariant set, and the results compared with~\cite{zhao2021synthesizing}. In Fig~\ref{fig:path-following}, UIS has a larger invariant set than one linear $\alpha$, and it is also comparable to the SyntheBC approach as described in~\cite{zhao2021synthesizing}. 
\end{example}
\section{Conclusion}
This paper presents a systematic method for deriving $\PWA$ BFs and their corresponding invariant sets for dynamical systems identified using ReLU NNs or their equivalent $\PWA$ dynamical systems. A key innovation of this method is the introduction of the leaky ReLU function as a simplified substitute to the complex $\mathcal{K}_\infty$ function typically used in BF formulations. 
As an extension of our previous work, we propose a novel approach, the Union of Invariant Sets (UIS), which takes advantage of all information obtained from previous methods to compute the largest possible $\PWA$ invariant set. The efficacy of the UIS framework has been demonstrated through a number of nontrivial examples. 
\begin{figure}
    \centering
    \includegraphics[width=0.5\linewidth]{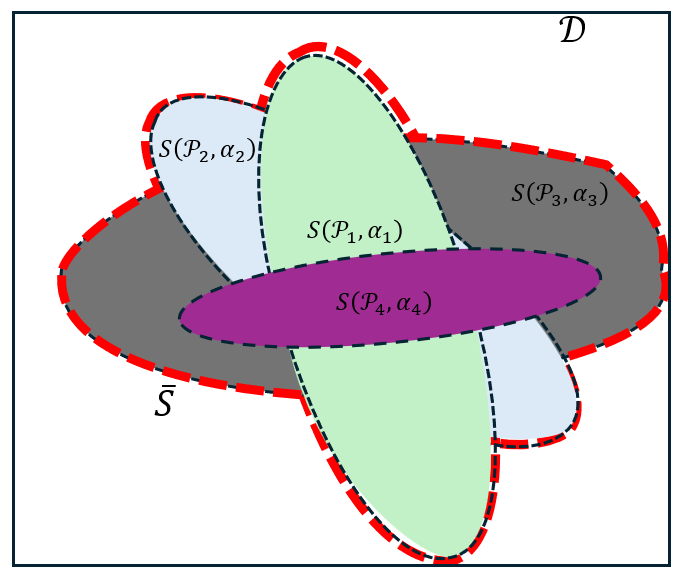}
    \caption{The new invariant set $\overline{S}$ shown with a red dashed line is the union of four invariant sets obtained from optimization problem~\eqref{eq:opt UIS} with different $\alpha$. As can be seen $S(\prtition_4,\alpha_4)$ does not contribute to the $\overline{S}$.}
    \label{fig:UIS process}
\end{figure}
\begin{figure}
    \centering
    \includegraphics[width=0.8\linewidth]{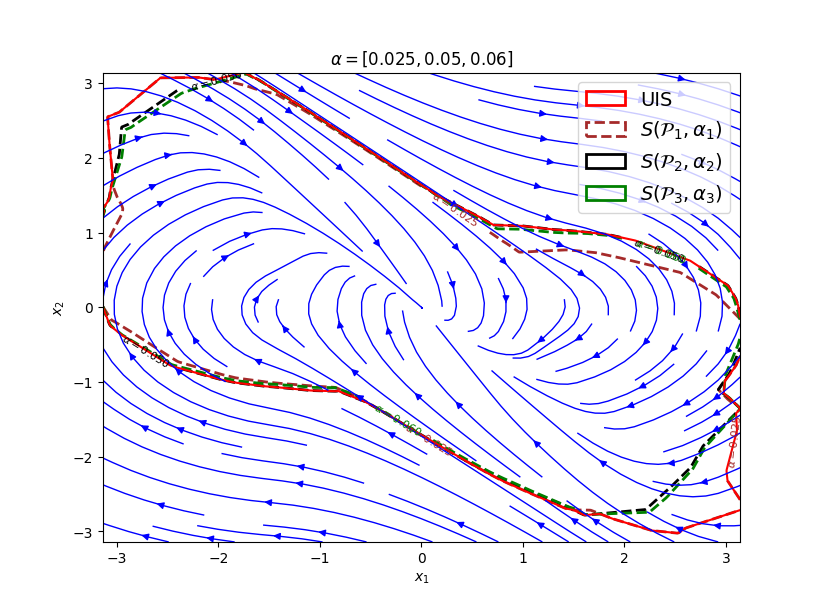}
    \caption{The final invariant set will be the Union of invariant sets with 3 different $\alpha$ as described in UIS~\ref{sec:UIS}.}
    \label{fig:UIS}
\end{figure}
\begin{figure}
    \centering
    \includegraphics[width=0.8\linewidth]{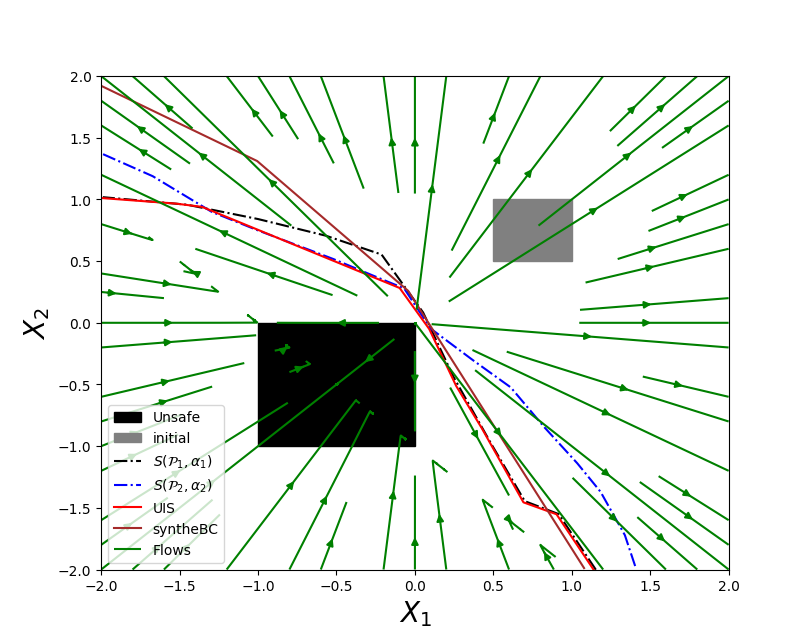}
    \caption{The final invariant set will be the Union of invariant sets with 2 different $\alpha$ as described in UIS~\ref{sec:UIS} for the Example~\ref{comparison example}.}
    \label{fig:path-following}
\end{figure}
\bibliographystyle{ieeetr}
\bibliography{acc2025}
\end{document}